%% file: main_v3.tex
\documentclass[reprint,
 amsmath,amssymb,
 aps,
]{revtex4-2}

\usepackage{graphicx}
\usepackage{dcolumn}
\usepackage{bm}
\usepackage{amsmath}
\usepackage{amsfonts}
\usepackage{amscd}
\usepackage{amsthm}
\usepackage{amssymb}

\newtheorem{proposition}{Proposition}
\newtheorem{corollary}{Corollary}

\usepackage{graphicx}
\usepackage{mathtools}
\usepackage{braket}
\usepackage{hyperref}
\hypersetup{
    colorlinks=true,       
    linkcolor=cyan,          
    citecolor=magenta,        
    filecolor=magenta,      
    urlcolor=cyan,           
    runcolor=cyan
}

\begin{document}

\preprint{APS/123-QED}

\title{Convergence condition of simulated quantum annealing for closed and open systems}

\author{Yusuke Kimura}
 \affiliation{International Research Frontiers Initiative, Tokyo Institute of Technology, Shibaura, Minato-ku, Tokyo 108-0023, Japan}

\author{Hidetoshi Nishimori} 
\affiliation{International Research Frontiers Initiative, Tokyo Institute of Technology, Shibaura, Minato-ku, Tokyo 108-0023, Japan} 
\affiliation{%
Graduate School of Information Sciences, Tohoku University, Sendai 980-8579, Japan 
}%
\affiliation{
RIKEN, Interdisciplinary Theoretical and Mathematical Sciences (iTHEMS), Wako, Saitama 351-0198, Japan
}%


\begin{abstract}
Simulated quantum annealing is a generic classical protocol to simulate some aspects of quantum annealing and is sometimes regarded as a classical alternative to quantum annealing in finding the ground state of a classical Ising model. We derive a generic condition for simulated quantum annealing to converge to thermal equilibrium at a given, typically low, temperature. 
Both closed and open systems are treated. We rewrite the classical master equation for simulated quantum annealing into an imaginary-time Schr\"odinger equation, to which we apply the imaginary-time variant of asymptotic adiabatic condition to deduce the convergence condition. The result agrees qualitatively with a rigorous convergence condition of simulated quantum annealing for closed systems, which was derived from the theory of inhomogeneous Markov process. Also observed is qualitative agreement with a rigorous convergence condition of quantum annealing for closed systems under the real-time Schr\"odinger dynamics. This coincidence of convergence conditions for classical stochastic processes for simulated quantum annealing and the real-time quantum dynamics for quantum annealing is highly non-trivial and calls for further scrutiny.
\end{abstract}

\maketitle

\section{Introduction}

Quantum annealing (QA) is a metaheuristic, a generic approximate algorithm, for combinatorial optimization problems \cite{Kadowaki1998,Farhi2001,Santoro2002,Das2008,Morita2008,Albash2018,Hauke2020,crosson2021}.  Recent years have seen its developments toward applications to a broader class of problems beyond combinatorial optimization such as quantum simulations \cite{Harris2018,King2018,Nishimura2020,Gardas2018,Bando2020,Weinberg2020,king2022a,King2022b} and optimization with continuous variables \cite{Abel2021a,Abel2021b,Abel2021c,Koh2022}.

The classical algorithm of simulated quantum annealing (SQA) has often been used to simulate some properties of quantum annealing \cite{Kadowaki1998t,Santoro2002,Heim2015} and shares the ultimate goal with QA to find the ground state of a classical Ising model. SQA uses stochastic processes to statistically sample the finite, typically very low (and ideally zero), temperature equilibrium state of the transverse-field Ising model written as a classical Ising model by the Suzuki-Trotter formula \cite{Suzuki}.  SQA is not designed to reproduce dynamical properties of QA as the former follows classical stochastic processes whereas the latter is realized by the time-dependent Schr\"odinger equation, which are  fundamentally different processes. 
This point is clear also from the fact that the Suzuki-Trotter formula is derived by rewriting the equilibrium partition function of a quantum system into the equilibrium partition function of a corresponding classical Ising model.

Low-temperature equilibrium behaviors of the transverse-field Ising model are expected to be captured by SQA as long as the rate of change of time-dependent parameters in the Hamiltonian is sufficiently slow because then the system follows quasi-equilibrium state. It is therefore natural to ask how slow is slow enough for this expectation to stay valid. Equivalently, we may ask, under what conditions on the time dependence of parameters, SQA successfully reaches thermal equilibrium.

This question has indeed been answered for generic closed systems, i.e., systems isolated from the environment
\footnote{As will be discussed later, the term ``environment'' here does not refer to classical thermal effects used in equilibrium simulations of SQA but means experimentally-relevant quantum-mechanical dissipation coming from interactions with a large number of degrees of freedom outside of qubit systems \cite{Caldeira1983}}, where the theory of inhomogeneous Markov process was used to analyze the stochastic process under SQA \cite{Morita2006,Morita2008}. It has been shown that a generic condition for convergence to a finite-temperature equilibrium state looks qualitatively similar to the corresponding generic convergence condition of QA for the transverse-field Ising model with time-dependent coefficients evolving under the real-time Schr\"odinger dynamics.  More precisely, in both cases of QA and SQA, the coefficient of the transverse-field term is to be decreased by a power law of time in the long-time limit with the power inversely proportionally to the system size. This apparent qualitative coincidence of convergence conditions for completely different dynamics is highly non-trivial, and it should be useful to provide a different perspective on this problem.  Also from the technical point of view, the proof provided in Refs.~\cite{Morita2006,Morita2008} is rather abstract and it is not easy to extract relevant physics out of it.

To mention a few related studies, Refs.~\cite{Albash_2012, Venuti2017} studied the quantum dynamics of open systems and Ref.~\cite{crosson2020} discussed the relaxation time of time-independent models. Though not directly relevant to the present study, those papers represent interesting developments from the general perspective of dynamical properties of QA or SQA.

The present paper represents part of efforts toward the above-mentioned goal concerning convergence conditions by deriving a condition of convergence of SQA for closed and open systems, based on an approximate, but asymptotically accurate, adiabatic condition for the imaginary-time Schr\"odinger equation. The classical stochastic dynamics of SQA following the master equation is transformed to the imaginary-time Schr\"odinger equation through a well-established prescription \cite{Nishimori2014} (see also Refs.~\cite{Henley2004,Castelnovo2005}). Then the leading-order term of the asymptotic expansion of the adiabatic condition for the imaginary-time  Schr\"odinger dynamics is used to derive our conclusion on the rate of time development of a coefficient. The logic is not completely rigorous mathematically because higher-order terms of the asymptotic expansion are ignored. 
Nevertheless, it turns out that the resulting convergence condition is in qualitative agreement with the corresponding rigorous condition derived from the theory of inhomogeneous Markov process for closed systems \cite{Morita2006}.

This paper is structured as follows: In Sec.~\ref{sec2}, we derive the imaginary-time Schr\"odinger equation from the classical master equation for the time-dependent transverse-field Ising model. In Sec.~\ref{sec3}, we deduce a convergence condition for SQA to a thermal equilibrium state by applying the imaginary-time variant of adiabatic condition. Based on the result in Sec.~\ref{sec3}, we derive a condition on the coefficient of the transverse field that ensures the convergence to a thermal equilibrium state  in Sec.~\ref{sec4}. This is a central result of this paper. We discuss the result and conclude the paper in Sec.~\ref{sec5}.

\section{Simulated quantum annealing and imaginary-time Schr\"odinger equation}
\label{sec2}

In this section, we recapitulate the general process to rewrite the classical master equation for the Ising model into an imaginary-time Schr\"odinger equation  under the context of SQA, closely following Ref.~\cite{Nishimori2014}.

\subsection{Classical master equation}
Our starting point is the master equation for the Markovian dynamics of the Ising model \cite{Nishimori2011book}:
\begin{equation}
\label{master eq in sec1}
\frac{d P_{\sigma}(t)}{dt}=\sum_{\sigma'} W_{\sigma\sigma'}P_{\sigma'}(t).
\end{equation}
Here, $\sigma$ represents a configuration of $N$ Ising spins,  $\sigma=\{\sigma_1, \sigma_2, \ldots, \sigma_N\}$, and $P_\sigma(t)$ is used to denote the probability of the system in configuration $\sigma$ at time $t$. We use the notation $\hat{W}$ to denote the transition matrix of size $2^N$ whose entries are: $(\hat{W})_{\sigma\sigma'}=W_{\sigma\sigma'}$. An off-diagonal entry of the matrix $\hat{W}$ can be written as \cite{Nishimori2014}:
\begin{equation}
W_{\sigma\sigma'}= w_{\sigma\sigma'} e^{-\frac{1}{2}\beta\left(H_0(\sigma)- H_0(\sigma') \right)} \hspace{10mm} (\sigma\ne \sigma'),
\end{equation}
where $\beta$ is the inverse temperature and $H_0(\sigma)$ is the Hamiltonian of the Ising model under consideration.
We employ the heat bath method with single-spin flips in this paper, where $w_{\sigma\sigma'}$ is given by
\begin{equation}
w_{\sigma\sigma'}=\frac{1}{ e^{\frac{1}{2}\beta\left(H_0(\sigma')- H_0(\sigma) \right)}+e^{-\frac{1}{2}\beta\left(H_0(\sigma')- H_0(\sigma) \right)}},
\end{equation}
where $\sigma$ and $\sigma'$ differ only at a single spin. We set $W_{\sigma\sigma'}=0$ for other types of transitions such as multi-spin flips.
Owing to the probability conservation condition $\sum_{\sigma'}W_{\sigma'\sigma}=0$, a diagonal entry of the matrix $\hat{W}$, $W_{\sigma\sigma}$, is expressed as
\begin{equation}
W_{\sigma\sigma}=-\sum_{\sigma'(\ne\sigma)}W_{\sigma'\sigma}=-\sum_{\sigma'(\ne\sigma)} e^{-\frac{1}{2}\beta\left(H_0(\sigma')- H_0(\sigma) \right)}\, w_{\sigma'\sigma}.
\end{equation}

\subsection{Imaginary-time Schr\"odinger equation for SQA}

For simplicity of presentation, we first analyze the case of closed systems, i.e., systems described solely by the classical representation of the transverse-field Ising model at finite temperature to be defined below.
Generalization to open systems is relatively straightforward and will be discussed in a later section.

The Hamiltonian of the transverse-field Ising model for QA is 
\begin{equation}
\label{eq:Ising}
H(t)=H_{\rm Ising}+H_{\rm TF}(t),
\end{equation}
where $H_{\rm Ising}$ denotes the Hamiltonian of the Ising model
\begin{align}
\label{eq:classical_Ising}
    H_{\rm Ising}=-\sum_{\langle jj'\rangle}J_{jj'}\,\sigma_j^z\sigma_{j'}^z
\end{align}
with $\langle jj'\rangle$ for pairs of interacting spins, and $H_{\rm TF}(t)$ is for the transverse field with a time-dependent coefficient:
\begin{align}
\label{eq:tf}
    H_{\rm TF}(t)=-\Gamma (t)\sum_{j=1}^N\sigma_j^x.
\end{align}
Here $\sigma_j^z$ and $\sigma_j^x$ are the $z$ and $x$ components of Pauli matrix at site $j$, respectively.

SQA runs Markov-chain Monte Carlo steps for the transverse-field Ising model rewritten as the corresponding classical Ising model derived by
the Suzuki-Trotter formula \cite{Suzuki} with the Hamiltonian:
\begin{equation}
\begin{split}
\label{Hamiltonian H0 in sec1}
& \beta H_0(\sigma)= \\ 
& -\sum_{k=1}^M\left(\sum_{\langle jj'\rangle}\frac{\beta J_{jj'}}{M} \sigma_j^{(k)}\sigma_{j'}^{(k)}+\sum_{j=1}^N \gamma(t) \sigma_j^{(k)}\sigma_j^{(k+1)}\right),
\end{split}
\end{equation}
where $M$ is the total number of Trotter slices running from $k=1$ to $M$, and the time dependent coefficient of the second term $\gamma(t)$ is given by
\begin{equation}
\label{def gamma in sec1}
\gamma(t)=\frac{1}{2} \log\left( {\rm coth} \frac{\beta \Gamma(t)}{M} \right).
\end{equation}
In Eq.~\eqref{Hamiltonian H0 in sec1},
$\sigma_j^{(k)}(=\pm 1)$ is used to denote a classical Ising spin at the $j$th site on the $k$th Trotter slice.
Notice that Eq.~\eqref{Hamiltonian H0 in sec1} has been derived with $\Gamma(t)$ regarded as a fixed parameter as can be understood from the fact that the Suzuki-Trotter formula treats the equilibrium partition function.

In SQA, $\beta$ and $M$ are both kept constant at large values, typically choosing $\beta/M$ to be of order unity, in order to reproduce the zero-temperature properties of the original transverse-field Ising model as faithfully as possible. Only the coefficient $\Gamma (t)$ (and thus $\gamma (t)$) changes with time $t$. 
It is important to notice here that the dynamics of SQA with time-dependent parameters is driven by the classical stochastic process of Monte Carlo simulation, which is completely different from the original (zero-temperature) quantum dynamics of the Schr\"odinger equation for the transverse-field Ising model of Eq.~\eqref{eq:Ising}. One should not expect that the dynamical, i.e. transient, properties of the original quantum system of Eq.~\eqref{eq:Ising} can be reproduced by the classical method of SQA \cite{Bando2021}. Nevertheless, QA and SQA share the same goal to find, in the long-time limit $t\to\infty$, the ground state of the classical Ising model of Eq.~\eqref{eq:classical_Ising} by tuning the parameter $\Gamma (t)$ appropriately as a function of time $t$.
The following discussions are for a condition for SQA with a time-dependent parameter to converge to a thermal equilibrium state at given finite values of $\beta$ and $M$ in the long-time limit, which is different from a thermal relaxation process in  a system with time-independent parameters.

The results deduced in this paper hold for any (real and positive) value of $\beta$. However, for practical applications of the results to a closed system, we are mainly interested in the limit at which $\beta$ is chosen very large.

Let us rewrite the classical master equation for SQA in terms of an imaginary-time Schr\"odinger equation. 
The quantum Hamiltonian dictating the latter equation, $\hat{H}$, can be constructed from the transition matrix $\hat{W}$ as \cite{Nishimori2014}:
\begin{equation}
\label{quantum Hamiltonian in sec1}
\hat{H}:=-e^{\frac{\beta \hat{H}_0}{2}} \hat{W} e^{\frac{-\beta \hat{H}_0}{2}}.
\end{equation}
In Eq.~\eqref{quantum Hamiltonian in sec1}, $\hat{H}_0$ is used to represent a diagonal matrix whose diagonal entries are: $(\hat{H}_0)_{\sigma\sigma}=H_0(\sigma)$. 

The quantum Hamiltonian has the following explicit expression:
\begin{equation}
\label{Hamiltonian explicit in sec1}
\hat{H}=\sum_{\sigma}\sum_{\sigma'} w_{\sigma\sigma'} \left( e^{-\frac{1}{2}\beta\left(H_0(\sigma')- H_0(\sigma) \right)}\ket{\sigma}\bra{\sigma}-\ket{\sigma'}\bra{\sigma} \right).
\end{equation} 
Since the configurations $\sigma$ and $\sigma'$ differ only at a single site, the $j$th site on the $k$th Trotter slice, the difference $\beta H_0(\sigma')- \beta H_0(\sigma)$ is given by the following expression, owing to the expression \eqref{Hamiltonian H0 in sec1} of the Hamiltonian $H_0$:
\begin{eqnarray}
\label{difference in sec1}
\beta H_0(\sigma')- \beta H_0(\sigma)= && \frac{2\beta}{M}\sum_{j' \, ({\rm n.n.}\, j)}J_{jj'}\sigma_j^{(k)}\sigma_{j'}^{(k)} \\ \nonumber
&&+2\gamma(t)\sigma_j^{(k)} \left( \sigma_{j}^{(k+1)}+ \sigma_{j}^{(k-1)}\right).
\end{eqnarray}
The summation symbol on the right-hand side expresses the sum over sites $j'$ that are interacting with $j$. For notational simplicity, we denote the difference \eqref{difference in sec1} by $-2\beta H_{j,k}$:
\begin{align}
\label{def Hjk in sec1}
&-\beta H_{j,k}:=\nonumber\\
&\frac{\beta}{M}\sum_{j' \, ({\rm n.n.}\, j)}J_{jj'}\sigma_j^{(k)}\sigma_{j'}^{(k)}+\gamma(t)\sigma_j^{(k)} \left( \sigma_{j}^{(k+1)}+ \sigma_{j}^{(k-1)}\right). 
\end{align}

Utilizing the introduced notation, the diagonal coefficient in Eq.~\eqref{Hamiltonian explicit in sec1} can be rewritten as:
\begin{equation}
\label{diagonal in sec1}
w_{\sigma\sigma'} e^{-\frac{1}{2}\beta\left(H_0(\sigma')- H_0(\sigma) \right)}=\frac{e^{\beta H_{j,k}}}{ e^{\beta H_{j,k}}+ e^{-\beta H_{j,k}}}.
\end{equation}
Similarly, the off-diagonal coefficient in Eq.~\eqref{Hamiltonian explicit in sec1} is:
\begin{equation}
\label{offdiagonal in sec1}
w_{\sigma\sigma'}=\frac{1}{ e^{\beta H_{j,k}}+ e^{-\beta H_{j,k}}}.
\end{equation}
Using a matrix-vector notation, the master equation \eqref{master eq in sec1} can be rewritten as follows:
\begin{equation}
\label{master eq matrix form in sec1}
\frac{d\hat{P}(t)}{dt} = \hat{W}(t) \hat{P}(t),
\end{equation}
where $\hat{P}(t)$ denotes a vector whose $\sigma$th entry $(\hat{P}(t))_{\sigma}$ is $P_\sigma(t)$: $(\hat{P}(t))_{\sigma}= P_\sigma(t)$.

If we define
\begin{align}
\label{state_correspondence}
    \ket{\phi(t)}:=e^{\frac{1}{2}\beta \hat{H}_0}\hat{P}(t),
\end{align}
it is straightforward to confirm that the master equation \eqref{master eq matrix form in sec1} is rewritten as the following imaginary-time Schr\"odinger equation \cite{Nishimori2014} :
\begin{equation}
\label{IT sch in sec1}
-\frac{d\ket{\phi(t)}}{dt}=\left( \hat{H}(t)-\frac{1}{2}\frac{d}{dt}(\beta \hat{H}_0) \right)\, \ket{\phi(t)}.
\end{equation}
In the long-time limit, time dependence in the Hamiltonian $\hat{H_0}$ becomes arbitrarily small as $\Gamma (t)$ gradually approaches zero, and the expression on the outer parentheses in the above equation approaches $\hat{H}(t)$.
In this limit, the ground state of the Hamiltonian $\hat{H}(t)$ corresponds to the equilibrium state of the master equation \eqref{master eq matrix form in sec1} as shown in Ref.~\cite{Nishimori2014}.

We apply the adiabatic condition \cite{Morita2008} to the imaginary-time Schr\"odinger equation of Eq.~\eqref{IT sch in sec1} to deduce the convergence condition on SQA.

\section{Imaginary-time variant of adiabatic condition}
\label{sec3}
We next discuss the adiabatic condition for the imaginary-time Schr\"odinger equation.

\subsection{General adiabatic condition for the imaginary-time Sch\"odinger equation}

To formulate the adiabatic condition, it is useful first to keep computation time to a finite value denoted as $\tau$. Then, $s$ is used for a dimensionless time scaled by $\tau$, i.e. $s=t/\tau$. We write the Hamiltonian of a generic quantum system as $\mathcal{H}(t)$. 
We have in mind the quantity in the large parentheses on the right-hand side of Eq.~\eqref{IT sch in sec1} for $\mathcal{H}(t)$ as will be discussed later.
We assume that the instantaneous ground state of $\mathcal{H}(t)$ is non-degenerate. This is justified in the present system, which is time-evolving by Eq.~\eqref{IT sch in sec1}, because the quantity in the large parentheses on the right-hand side of Eq.~\eqref{IT sch in sec1} is the Hamiltonian of a transverse-field Ising model 
with non-vanishing off-diagonal elements only between those states different by a single-spin flip (and thus expressed by $\sigma_i^x$), to which the Perron-Frobenius theorem \cite{Pillai2005} applies 
\footnote{
Also in the long-time limit $t\to\infty$, the ground state is non-degenerate since, according to Eq.~\eqref{state_correspondence}, the ground state of $\mathcal{H}(t)$ corresponds to the equilibrium Boltzmann distribution of the original classical system \cite{Nishimori2014}, which is naturally unique. This should not be confused with the $Z_2$ degeneracy of the ground state of the classical Ising model of Eq.~\eqref{eq:classical_Ising}.}.
Notice that this transverse-field Ising model defined in terms of $\hat{H}$ of Eq.~\eqref{Hamiltonian explicit in sec1} and $d(\beta \hat{H_0})/dt$ is different from the original model Eq.~\eqref{eq:Ising}.

We define $\overline{\mathcal{H}}(s)$ as 
\begin{equation}
\overline{\mathcal{H}}(s) := \mathcal{H}(t).
\end{equation}

Let $\ket{j(s)}$ denote the $j$th excited state, $j=0,1, \ldots$,
\begin{equation}
\overline{\mathcal{H}}(s)\ket{j(s)}=\epsilon_j(s)\ket{j(s)}.
\end{equation}
$\Delta_j(s)$ is used to denote the instantaneous energy gap between the $j$th excited state and the ground state, $\Delta_j(s)=\epsilon_j(s)-\epsilon_0(s)$. The instantaneous energy gap between the first excited state and the ground state is particularly written as $\Delta(s)$, i.e. $\Delta(s)=\epsilon_1(s)-\epsilon_0(s)$.

A variant of the adiabatic condition for the imaginary-time Schr\"odinger equation obtained in Ref.~\cite{Morita2008} states that: For the state vector 
\begin{align}
\ket{\psi(s)}=\sum_j c_j(s)\ket{j(s)}
\end{align}
of the imaginary-time Schr\"odinger equation, 
\begin{align}
    -\frac{1}{\tau}\frac{d\ket{\psi(s)}}{ds}=\overline{\mathcal{H}}(s)\ket{\psi(s)},
\end{align}
the coefficient $c_j(s)$ with $j\ne 0$ has the large-$\tau$ asymptotic expansion as:
\begin{equation}
\label{ineq1}
c_{j\ne 0}(s) \approx \frac{1}{\tau} \frac{\bra{j(s)}\frac{d\overline{\mathcal{H}}(s)}{ds}\ket{0(s)}}{\Delta_j(s)^2} +O(\tau^{-2}).
\end{equation}
The first term on the right-hand side of Eq.~\eqref{ineq1}, the first-order term in the asymptotic expansion from the limit of large $\tau$, should be sufficiently small for the system to be close to the instantaneous ground state $\ket{0(s)}$.
Since we ignore the second- and higher-order terms, our theory is approximate but is nevertheless expected to be accurate in the asymptotic limit of large $\tau$.

In terms of the original time variable $t$, the adiabatic condition can be written as:
\begin{equation}
\label{ineq2}
\frac{\lvert \bra{j(t)}\frac{d\mathcal{H}(t)}{dt}\ket{0(t)}\rvert} {\Delta_j(t)^2}\ll 1
\end{equation}
for $j\ne 0$.
By bounding the numerator of the left-hand side of Eq.~\eqref{ineq2} by the operator norm as shown in Appendix A,
\begin{align}
\label{eq:norm_inequality}
    \left\lvert \bra{j(t)}\frac{d\mathcal{H}(t)}{dt}\ket{0(t)}\right\rvert \le\left\lVert \frac{d\mathcal{H}(t)}{dt}\right\rVert,
\end{align}
we find that the condition is satisfied that the original system reaches thermal equilibrium (the ground state of the imaginary-time Schr\"odinger equation) in the long-time limit if the following relation is satisfied:
\begin{equation}
\label{it-adiabaticity}
\lim_{t\to\infty} \frac{\left\lVert \frac{d\mathcal{H}(t)}{dt}\right\rVert}{\Delta(t)^2}\ll 1.
\end{equation}

Notice that we apply the adiabatic condition only in the long-time limit since we are interested in minimizing the excitation probability only at the end of the process such that the system approaches the (low-temperature) equilibrium state in the long-time limit. It may therefore happen that the adiabatic condition may not necessarily be satisfied in the intermediate time region, which we do not care.

\subsection{Adiabatic condition applied to SQA}

Let us apply the condition of Eq.~\eqref{it-adiabaticity} to the problem we are interested in.

We apply the framework in the previous subsection to the case that $\mathcal{H}(t)$ is the quantity in the large parentheses on right-hand side of the imaginary-time Schr\"odinger equation \eqref{IT sch in sec1}: 
\begin{equation}
\mathcal{\hat{H}}(t)=\hat{H}(t)-\frac{1}{2}\frac{d}{dt}(\beta \hat{H}_0).
\end{equation}
Notice that $\mathcal{H}(t)$ in the previous subsection is for a generic Hamiltonian, and $\hat{\mathcal{H}}(t)$ here is the above specific one.
According to Eq.~\eqref{it-adiabaticity}, convergence of SQA to the equilibrium state is assured in the long-time limit when the following condition is satisfied:
\begin{equation}
\label{adiabatic}
\lim_{t\to\infty} \frac{\left\lVert\frac{d\mathcal{\hat{H}}(t)}{dt}\right\rVert}{\Delta(t)^2} \ll 1.
\end{equation} 

It is possible to yield an upper bound on the norm in the above numerator using Eq.~\eqref{Hamiltonian explicit in sec1}. Equation~\eqref{diagonal in sec1} together with Eq.~\eqref{difference in sec1} implies that the diagonal term can be written as a finite number of terms of interacting spins. The off-diagonal term is just a transverse-field term with an appropriate coefficient evaluated by Eq.~\eqref{offdiagonal in sec1}.
We thus find:
\begin{equation}
\begin{split}
\label{ineq tilde H}
& \left\lVert\frac{d\mathcal{\hat{H}}(t)}{dt}\right\rVert \le \\ 
& bMN \left(\Bigg\lVert\frac{d}{dt}(w_{\sigma\sigma'} e^{-\frac{1}{2}\beta\left(H_0(\sigma')- H_0(\sigma) \right)})\Bigg\rVert+\Bigg\lVert \frac{d}{dt} w_{\sigma\sigma'} \Bigg\rVert\right)\\ 
& +\frac{1}{2}\Bigg\lVert \frac{d^2}{dt^2}(\beta \hat{H}_0) \Bigg\rVert,
\end{split}
\end{equation}
where $b$ is a constant independent of $M, N, \beta$, and $t$ but is determined by the number of sites interacting with a given site.
Utilizing the relation
\begin{equation}
\Big\lVert \frac{d}{dt}(\beta H_{j,k}) \Big\rVert=\Big\lVert \gamma'(t)\sigma_j^{(k)} \left( \sigma_{j}^{(k+1)}+ \sigma_{j}^{(k-1)}\right) \Big\rVert \le 2\lvert \gamma'(t)\rvert, 
\end{equation}
we have the following bound on the first term on the right-hand side of the inequality \eqref{ineq tilde H}:
\begin{align}
& \Big\lVert \frac{d}{dt}(w_{\sigma\sigma'} e^{-\frac{1}{2}\beta\left(H_0(\sigma')- H_0(\sigma) \right)}) \Big\rVert \nonumber\\
&= \Bigg\lVert \frac{d}{dt}\left(\frac{e^{\beta H_{j,k}}}{ e^{\beta H_{j,k}}+ e^{-\beta H_{j,k}}}\right) \Bigg\rVert \nonumber\\
& = \frac{1}{2}\Bigg\lVert \frac{d}{dt}(\beta H_{j,k})\cdot {\rm sech}^2(\beta H_{j,k}) \Bigg\rVert\le \lvert \gamma'(t)\rvert.
\end{align}
An analogous computation shows that
\begin{equation}
\label{ineq4} 
\begin{aligned}
& \Bigg\lVert \frac{d}{dt}w_{\sigma\sigma'} \Bigg\rVert = \Bigg\lVert \frac{d}{dt}\frac{1}{ e^{\beta H_{j,k}}+ e^{-\beta H_{j,k}}} \Bigg\rVert \\
& = \frac{1}{2}\Bigg\lVert \frac{d}{dt}(\beta H_{j,k})\cdot {\rm tanh}(\beta H_{j,k}){\rm sech}(\beta H_{j,k}) \Bigg\rVert \\
& \le \frac{1}{2}\lvert \gamma'(t)\rvert.
\end{aligned}
\end{equation}
We also have 
\begin{equation}
\Bigg\lVert \frac{d^2}{dt^2}(\beta \hat{H}_0) \Bigg\rVert\le \sum_{k=1}^M \sum_{j=1}^N \Bigg\lVert \gamma''(t)  \sigma_j^{(k)}\sigma_j^{(k+1)} \Bigg\rVert \le MN \lvert \gamma''(t)\rvert.
\end{equation}
These computations applied to Eq.~\eqref{ineq tilde H} yield the following bound:
\begin{equation}
\label{bound tilde H}
\left\lVert\frac{d\mathcal{\hat{H}}(t)}{dt}\right\rVert \le MN\left(\frac{3b}{2}\, \lvert \gamma'(t)\rvert+\frac{1}{2}\, \lvert \gamma''(t)\rvert\right).
\end{equation}

Finding a lower bound of the energy gap $\Delta(t)$ between the first excited state and the ground state is a nontrivial step. The coefficient of the transverse field in the quantum Hamiltonian \eqref{Hamiltonian explicit in sec1} is $w_{\sigma\sigma'}$. We start from the following inequality for the coefficient $w_{\sigma\sigma'}$:
\begin{equation}
w_{\sigma\sigma'}=\frac{1}{ e^{\beta H_{j,k}}+ e^{-\beta H_{j,k}}}\ge \frac{1}{ 2\, e^{\lvert\beta H_{j,k}\rvert}}.
\end{equation}
$\lvert\beta H_{j,k}\rvert$ has the following upper bound owing to the definition of $H_{j,k}$ in Eq.~\eqref{def Hjk in sec1}:
\begin{equation}
\lvert\beta H_{j,k}\rvert\le \frac{\beta}{M}\bigg\lvert\sum_{j' \, ({\rm n.n.}\, j)}J_{jj'}\sigma_j^{(k)}\sigma_{j'}^{(k)}\bigg\rvert+2\gamma(t).
\end{equation} 
We adopt the notation 
\begin{align}
\label{p-def}
p(M):={\rm max}_{j,k}\, \frac{1}{M}\bigg\lvert\sum_{j' \, ({\rm n.n.}\, j)}J_{jj'}\sigma_j^{(k)}\sigma_{j'}^{(k)}\bigg\rvert, 
\end{align}
then $\lvert\beta H_{j,k}\rvert$ has an upper bound
\begin{equation} 
\lvert\beta H_{j,k}\rvert\le \beta p(M)+2\gamma(t).
\end{equation}
From these relations, one learns that the coefficient $w_{\sigma\sigma'}$ has a lower bound as
\begin{equation}
\label{l bound w}
w_{\sigma\sigma'}\ge \frac{1}{ 2\, e^{\beta p(M)+2\gamma(t)}}.
\end{equation}

For a generic transverse-field Ising model, a lower bound of the energy gap $\Delta(t)$ between the ground state and the first excited stated was deduced in Refs.~\cite{Somma2007, Morita2007, Morita2008}. Applying the argument given in those references to the model under discussion, we obtain the following relation:
\begin{equation}
\Delta(t)\ge A(N)\, (w_{\sigma\sigma'})^N,
\end{equation}
where the $N$-dependence of $A(N)$ is discussed below. 
Making use of the lower bound for the coefficient $w_{\sigma\sigma'}$ in Eq.~\eqref{l bound w}, we deduce the lower bound for $\Delta(t)$
\begin{equation}
\Delta(t)\ge \frac{A(N)}{2^N}e^{-N(\beta p(M)+2\gamma(t))}.
\end{equation}
The coefficient $A(N)$ does not depend on time $t$ but depends on $N$ for $N\gg 1$ as:
\begin{equation}
A(N)=a \sqrt{N} e^{-cN},
\end{equation}
where $a$ and $c$ are positive constants independent of $N$ in the asymptotic limit of large $N$. Thus, we obtain the lower bound of the energy gap $\Delta(t)$ as follows:
\begin{equation}
\label{gap lower bound}
\Delta(t)\ge \frac{a \sqrt{N}}{2^N} e^{-N(\beta p(M)+2\gamma(t)+c)}.
\end{equation}

As a result of this equation and Eq.~\eqref{bound tilde H}, we find the following bound:
\begin{equation}
\begin{split}
\label{bound explicit}
&\lim_{t\to\infty} \frac{\left\lVert\frac{d\mathcal{\hat{H}}(t)}{dt}\right\rVert}{\Delta(t)^2} \le \\
&\lim_{t\to\infty} \frac{2^{2N} M\left(3b\, \lvert \gamma'(t)\rvert+\lvert \gamma''(t)\rvert\right) }{2a^2}
 \, e^{2N(\beta p(M)+2\gamma(t)+c)}.
\end{split}
\end{equation}
Dropping constants of order one, we learn from this result that the condition 
\begin{equation}
\label{bound_condition2}
\begin{split}
&\lim_{t\to\infty}\big(3b\, \lvert \gamma'(t)\rvert+\lvert \gamma''(t)\rvert\big) e^{4N\gamma(t)}\\
& \hspace{1cm} \times M\,2^{2N} e^{2N\beta p(M)+2Nc}\ll 1
\end{split}
\end{equation}
suffices to ensure the convergence of SQA for fixed values $N, M$, and $\beta$. Notice that the time-dependent part has been collected in the first line of the above equation.

\section{The result}
\label{sec4}
We study in this section how the condition of Eq.~\eqref{bound_condition2} leads to an explicit form of the coefficient $\gamma (t)$.
\subsection{Convergence condition for closed systems}
\label{sec4.1}
Here, we deduce a condition that ensures convergence of SQA from the bound \eqref{bound_condition2}.

It is natural to assume that $\Gamma (t)$ is a monotonically decreasing function from a large value at $t=0$ toward zero in the long-time limit. Thus $\gamma (t)$ is monotonically increasing according to the definition of Eq.~\eqref{def gamma in sec1}.
Suppose that the first term in the first line on the left-hand side of Eq.~\eqref{bound_condition2} is fixed to a small constant $c_1$.
Solving the differential equation 
\begin{equation}
\lvert \gamma'(t)\rvert e^{4N\gamma(t)}=c_1,
\end{equation}
we obtain the solution in terms of $\Gamma (t)$ as follows:
\begin{equation}
\label{eq:sol1_Gamma}
\Gamma(t) = \frac{M}{\beta} {\rm tanh}^{-1}\left( \frac{1}{(4N)^{1/2N} (c_1 t+c_2)^{1/2N}}\right),
\end{equation}
where $c_2$ is an integral constant.  With this form of $\Gamma(t)$, one finds that $\lvert \gamma''(t)\rvert e^{4N\gamma(t)}=c_1^2 (c_1 t+c_2)^{-1}$, which tends to 0 as $t\to\infty$; therefore the solution \eqref{eq:sol1_Gamma} satisfies the aymptotic adiabatic condition for the imaginary-time Schr\"odinger equation, given that $c_1$ is chosen small enough as required in Eq.~\eqref{bound_condition2}. 

This result inspires us to analyze a more general function of the following form for the coefficient of the transverse field $\Gamma(t)$:
\begin{equation}
\label{eq:Gamma_general}
\Gamma(t) = \frac{M}{\beta} {\rm tanh}^{-1}\left( \frac{1}{(c_1 t+c_2)^{g(t)}}\right).
\end{equation}
Here, $g(t)$ denotes a twice-differentiable, strictly positive function, $g(t)>0$, and $c_1$ and $c_2$ denote constants with $c_1>0$ as we mentioned previously. For $\Gamma(t)$ of the form \eqref{eq:Gamma_general} to ensure convergence, conditions are imposed on the functions $g(t)$, $g'(t)$, and $g''(t)$. We state this main result as the following proposition:

\begin{proposition}
\label{prop1}
Simulated quantum annealing (SQA) for closed systems converges to thermal equilibrium of the final Ising Hamiltonian in the long-time limit when the function $g(t)$ in $\Gamma(t)$ of Eq.~\eqref{eq:Gamma_general} satisfies the following conditions for sufficiently large $t$,
\begin{align}
    &0<g(t)\le \frac{1}{2N}
    \label{eq:condition_1},\\
    &
    \lvert g'(t)\rvert \le \frac{c'}{(c_1 t+c_2) \log(c_1 t+c_2)},
    \label{eq:condition_2}\\
    & \lvert g''(t)\rvert \le \frac{c''}{(c_1 t+c_2) \log(c_1 t+c_2)},
    \label{eq:condition_3}
\end{align}
where $c'$ and $c''$ denote positive constants, provided that the constants $c_1$, $c'$, and $c''$ are chosen small enough satisfying
\begin{align}
    \left(\frac{3b\,c_1}{4N}+\frac{3b\,c'}{2}+\frac{c''}{2}\right) M\, 2^{2N} e^{2N\beta p(M) +2Nc}\ll 1.
    \label{c-conditions}
\end{align}
for fixed, possibly large, values of $M, N$, and $\beta$.
\end{proposition}

\begin{proof}

First, we show that the conditions \eqref{eq:condition_1} and \eqref{eq:condition_2} are necessary for the QA convergence. When $\Gamma(t)$ takes the form \eqref{eq:Gamma_general}, $\gamma(t)$ is given as 
\begin{equation}
\label{eq for gamma}
\gamma(t)=\frac{g(t)}{2}\log\,(c_1 t+c_2).
\end{equation}
Then, $\gamma'(t) e^{4N\gamma(t)}$ has the following expression:
\begin{equation}
\begin{split}
\label{eq:1st_bound}
\gamma'(t) e^{4N\gamma(t)} &= \frac{g(t)}{2}c_1\,(c_1 t+c_2)^{2Ng(t)-1}\\ 
&+\frac{g'(t)}{2}(c_1 t+c_2)^{2Ng(t)}\log\,(c_1 t+c_2).
\end{split}
\end{equation}
The first term on the right-hand side of Eq.~\eqref{eq:1st_bound} needs to be bounded from above by a finite constant as $t\to\infty$, which requires $2Ng(t)-1\le 0$. This yields the first stated condition \eqref{eq:condition_1}. 
For a similar reason, the second term on the right-hand side of Eq.~\eqref{eq:1st_bound} needs to be bounded from above by a finite constant as $t\to\infty$, which yields the second condition \eqref{eq:condition_2}. 
When those conditions of Eqs.~\eqref{eq:condition_1} and \eqref{eq:condition_2} are satisfied, $\lvert \gamma'(t)\rvert e^{4N\gamma(t)}$ has the following bound:
\begin{eqnarray}
\lvert \gamma'(t)\rvert e^{4N\gamma(t)} 
\le \frac{c_1}{4N}+\frac{c'}{2}.
\end{eqnarray}
Thus, we find that the term $\lvert \gamma'(t)\rvert e^{4N\gamma(t)}$ can be made arbitrarily small when $c_1$ and $c'$ are chosen sufficiently small under the conditions \eqref{eq:condition_1} and \eqref{eq:condition_2}.

Next we explain that the condition \eqref{eq:condition_3} is also necessary to keep the term $\lvert\gamma''(t)\rvert e^{4N\gamma(t)}$ arbitrarily small as $t\to\infty$. 
With Eq.~\eqref{eq for gamma} for $\gamma (t)$, we obtain the following expression for $\gamma''(t) e^{4N\gamma(t)}$:
\begin{equation}
\begin{split}
\label{eq:2nd_bound}
\gamma''(t) e^{4N\gamma(t)}& = c_1 g'(t)\,(c_1 t+c_2)^{2Ng(t)-1}\\ 
& -c_1^2\frac{g(t)}{2}(c_1 t+c_2)^{2Ng(t)-2}\\ 
&+\frac{g''(t)}{2}(c_1 t+c_2)^{2Ng(t)}\log\,(c_1 t+c_2).
\end{split}
\end{equation}
For the third term on the right-hand side of Eq.~\eqref{eq:2nd_bound} to be bounded from above by a finite constant, we need the condition \eqref{eq:condition_3}.

One can confirm that the term $\lvert \gamma''(t)\rvert e^{4N\gamma(t)}$ can be made arbitrarily small as $t\to\infty$ under the conditions \eqref{eq:condition_1}, \eqref{eq:condition_2}, and \eqref{eq:condition_3}, when $c''$ is chosen sufficiently small. 
In fact, one has the following bound under the stated conditions:
\begin{equation}
\label{eq:eval_gamma2}
\begin{split}
&\lvert\gamma''(t)\rvert e^{4N\gamma(t)} \le c_1 \lvert g'(t)\rvert\,(c_1 t+c_2)^{2Ng(t)-1} \\
&+c_1^2\frac{g(t)}{2}(c_1 t+c_2)^{2Ng(t)-2}\\ 
&+\frac{\lvert g''(t)\rvert}{2}(c_1 t+c_2)^{2Ng(t)}\log\,(c_1 t+c_2)\\ 
&\le \frac{c_1 c'}{(c_1 t+c_2) \log(c_1 t+c_2)}+\frac{c_1^2}{4N}(c_1 t+c_2)^{-1}+\frac{c''}{2}.
\end{split}
\end{equation}
In the last line of this equation, the first two terms tend to 0 as $t\to\infty$; the last term becomes arbitrarily small when $c''$ is chosen small enough by following Eq.~\eqref{c-conditions}. This concludes the proof of the proposition.
\end{proof}

\noindent {\it Remark 1}.\\
Notice that $g(t)$ should increase as $t\to \infty$ if we demand $\Gamma(t)$ to decrease faster than a polynomial of $t$. Therefore, we learn that the coefficient $\Gamma(t)$ of the form \eqref{eq:Gamma_general} cannot decrease faster than a polynomial of $t$.  This result can be compared with a similar result deduced in Ref.~\cite{Kimura2022} for the real-time Schr\"odinger dynamics.
\vspace{3mm}

\noindent {\it Remark 2}.\\
The expression of Eq.~\eqref{eq:sol1_Gamma} for $\Gamma (t)$ turns out to be best possible according to Eq.~\eqref{eq:condition_1}.
\vspace{3mm}

\noindent
{\it Remark 3}.\\
By choosing $\beta$ and $M$ sufficiently large, the system would become sufficiently close to the ground state of the classical Ising model of Eq.~\eqref{eq:classical_Ising} in the long-time limit.

\subsection{Convergence condition for open systems}

The standard model of environmental effects for open systems is the following Hamiltonian with an additional term involving bosonic degrees of freedom \cite{Werner2005}:
\begin{align}
\label{eq:bosonic_environment}
    H(t)&=H_{\rm Ising}+H_{\rm TF}(t)\nonumber\\
    &+\sum_{j,l}\Big( c_l (a_{j,l}^{\dagger}+a_{j,l})\sigma_j^z +\omega_{j,l}a_{j,l}^{\dagger}a_{j,l}\Big),
\end{align}
where $c_l$ is the coupling strength of the $l$th oscillator to the spin system and $\omega_{j,l}$ is its frequency \footnote{The effect represented by the bosonic term in \eqref{eq:bosonic_environment} is different from the effect of temperature in SQA introduced by $\beta$ as mentioned in Ref.~[24]}. We assume that the Ohmic spectral function, $J(\omega)$, with coupling strength $\alpha$ and a cut-off frequency $\omega_c$ is given as 
\begin{align}
    J(\omega)=4\pi \sum_l c_l^2\delta(\omega-\omega_{j,l})=
    \begin{cases}
    2\pi \alpha \omega & (\omega< \omega_c)\\
    0 &(\omega\ge \omega_c)
    \end{cases}
    .
\end{align}
Then, tracing out the bosonic degrees of freedom in the Suzuki-Trotter formulation, we obtain the following expression representing environmental effects to be added to the original Trotterized Ising Hamiltonian $H_0$ \cite{Werner2005}: 
\begin{align}
\label{appended term}
-\frac{\alpha}{2}\left(\frac{\pi}{M}\right)^2\sum_{j=1}^N\sum_{k>k'}\left( {\rm sin} \frac{\pi \lvert k-k' \rvert}{M}\right)^{-2} \sigma_j^{(k)}\sigma_j^{(k')}.
\end{align}
Since this term is independent of time $t$, the theory developed for closed systems remains almost unchanged.
Including the term Eq.~\eqref{appended term} modifies $H_{j,k}$ in Eq.~\eqref{def Hjk in sec1} and $p(M)$ in Eq.~\eqref{p-def} as:
\begin{align}
\label{def Hjk modified}
&-\beta H_{j,k}=\nonumber\\
&\frac{\beta}{M}\sum_{j' \, ({\rm n.n.}\, j)}J_{jj'}\sigma_j^{(k)}\sigma_{j'}^{(k)}+\gamma(t)\sigma_j^{(k)} \left( \sigma_{j}^{(k+1)}+ \sigma_{j}^{(k-1)}\right)\nonumber\\
&+\frac{\alpha}{2}\left(\frac{\pi}{M}\right)^2\sum_{k'(\ne k)}\left( {\rm sin} \frac{\pi \lvert k-k' \rvert}{M}\right)^{-2} \sigma_j^{(k)}\sigma_j^{(k')} 
\end{align}
and 
\begin{align}
\label{p-def-open}
&p(M)={\rm max}_{j,k}\, \bigg\lvert\frac{1}{M}\sum_{j' \, ({\rm n.n.}\, j)}J_{jj'}\sigma_j^{(k)}\sigma_{j'}^{(k)}\nonumber\\
&+\frac{\alpha}{2\beta}\left(\frac{\pi}{M}\right)^2\sum_{k'(\ne k)}\left( {\rm sin} \frac{\pi \lvert k-k' \rvert}{M}\right)^{-2} \sigma_j^{(k)}\sigma_j^{(k')}\bigg\rvert, 
\end{align}
respectively.
Proposition 1 applies with those minor amendments.
 We summarize this observation as a Corollary.
\begin{corollary}
Proposition 1 applies to open systems with the parameter $p(M)$ in Eq.~\eqref{p-def-open}.
\end{corollary}

\subsection{Bounded coefficients}
The Schr\"odinger dynamics of the following form is often considered
\begin{equation}
\label{eq:bounded_Sch}
i\frac{d}{dt}\ket{\psi(t)}=\Big(s(t)\, H_{\rm Ising}-\left(1-s(t)\right)\sum_{j=1}^N \sigma^x_j\Big)\ket{\psi(t)},
\end{equation}
where $s(t)$ is a monotonically increasing function with $0\le s(t) \le 1$, in place of the Hamiltonian \eqref{eq:Ising}. We show that slight amendments of discussions of previous sections make it possible to apply the results, Proposition 1 and Corollary 1, to this case. Equations are written for real-time Schr\"odinger dynamics, which applies to the imaginary-time version as well.

One can rewrite Eq.~\eqref{eq:bounded_Sch} as 
\begin{equation}
\label{eq:bounded_Sch_2}
\frac{i}{s(t)}\frac{d}{dt}\ket{\psi(t)}= \Big(H_{\rm Ising}-\frac{1-s(t)}{s(t)} \sum_{j=1}^N \sigma_j^x\Big)\ket{\psi(t)}.
\end{equation}
We define a function $\tilde{t}$, that is a monotonic function of time $t$, as
\begin{equation}
\tilde{t}\equiv \int_0^t dt\, s(t).
\end{equation}
Using the introduced function $\tilde{t}$, one can rewrite Eq.~\eqref{eq:bounded_Sch_2} as follows:
\begin{equation}
\label{eq:bounded_Sch_3}
i\frac{d}{d\tilde{t}} \ket{\psi(t)} = \Big(H_{\rm Ising} - \frac{1-s(t)}{s(t)}\sum_{j=1}^N\sigma^x_j\Big)\ket{\psi(t)}.
\end{equation}

Now, we introduce a function $\Gamma(\tilde{t})$ defined as $\Gamma(\tilde{t})\equiv \left(1-s(t)\right)/s(t)$ to rewrite Eq.~ \eqref{eq:bounded_Sch_3} as
\begin{equation}
\label{eq:bounded_Sch_4}
i\frac{d}{d\tilde{t}} \ket{\psi(t)} =\Big( H_{\rm Ising} - \Gamma(\tilde{t}) \sum_{j=1}^N\sigma^x_j\Big)\ket{\psi(t)}.
\end{equation}
The arguments given in the previous sections applied to Eq.~\eqref{eq:bounded_Sch_4} reveal that, in the limit $\tilde{t}\to\infty$, $\Gamma(\tilde{t})$ proportional to $c_1 \tilde{t}^{-\tilde{g}(\tilde{t})}$,
\begin{equation}
\label{eq:exp_s}
\Gamma(\tilde{t})=\frac{1-s(t)}{s(t)}\propto (c_1\tilde{t})^{-\tilde{g}(\tilde{t})},
\end{equation}
where $\tilde{g}(\tilde{t})$ is subject to the conditions \eqref{eq:condition_1}, \eqref{eq:condition_2}, and \eqref{eq:condition_3}, ensures convergence. We used the fact that \begin{align}
\Gamma(t)=\frac{M}{\beta} {\rm tanh}^{-1}\left(\frac{1}{(c_1 t+c_2)^{g(t)}}\right)\sim \frac{M}{\beta}(c_1 t+c_2)^{-g(t)}
\end{align}
when $t$ is sufficiently large, and $c_2$ is suppressed in the large $t$ limit. 

An expression for $s(t)$ is obtained from Eq.~ \eqref{eq:exp_s} as
\begin{equation}
s(t)=\frac{1}{1+c_1\tilde{t}^{-\tilde{g}(\tilde{t})}} \approx 1-(c_1\tilde{t})^{-\tilde{g}(\tilde{t})}~~~ (\tilde{t}\gg 1).
\end{equation}

For large $t$, in some situations, $\tilde{t}$ can be approximated by $t$. For example, one can choose $s(t)={\rm tanh}\, t$, then $\tilde{t} = \log({\rm cosh}\, t) \approx t$ for sufficiently large $t$ ($t\gg 1$). 

\section{Discussion}
\label{sec5}
We have derived a condition that simulated quantum annealing (SQA) converges to thermal equilibrium for generic closed and open systems using the approximate, but asymptotically correct, adiabatic condition for the imaginary-time Schr\"odinger equation, to which the classical master equation governing SQA has been reduced. The result is in qualitative agreement with the rigorous version of convergence condition for closed systems derived from the theory of inhomogeneous Markov process \cite{Morita2006,Morita2008}, where the condition of convergence has been proved to be:
\begin{align}
\label{eq:inhomogeneous_result}
\Gamma (t) \ge \frac{M}{\beta} \tanh^{-1}\frac{1}{(t+2)^{2/RL_1}}.
\end{align}
Here $R$ is a constant of order $N$ and $L_1$ is an $N$-independent constant. Notice that the unit of time is different between those approaches: discrete steps for the above result in Refs.~\cite{Morita2006,Morita2008} and continuous time evolution in the present paper, resulting in the difference of coefficients of $t$, unity in Eq.~\eqref{eq:inhomogeneous_result} and $c_1$ in Eq.~\eqref{eq:Gamma_general}. It is anyway encouraging that essentially the same conclusion has been reached by two completely different methods, partly because the present, physics-oriented, approach may be more flexible in applications to other problems  than using the mathematically-solid, yet hard to generalize, theory of inhomogeneous Markov process.

It is worth attention that there is no essential difference in the convergence conditions for closed and open systems. The origin of this fact is that the additional term for environmental effects in Eq.~\eqref{appended term} is time independent. We thus conclude that phase-flip errors caused by the environment as represented by Eq.~\eqref{eq:bosonic_environment} does not essentially modify the generic bound on the rate of change of the time-dependent coefficient.

Another interesting fact is that the time dependence of the coefficient $\Gamma (t)$ of Eq.~\eqref{eq:sol1_Gamma} is quite similar to the corresponding convergence condition for real-time QA for closed systems as derived earlier under rigorous \cite{Kimura2022} and asymptotic \cite{Somma2007,Morita2008,Morita2007} adiabatic conditions. 
For example, the asymptotic version of convergence condition under the real-time Schr\"dinger dynamics is
\begin{align}
\label{polyn_power}
    \Gamma (t)=a (\delta t +c)^{-1/(2N-1)},
\end{align}
where $a, \delta$, and $c$ are time-independent constants. On the side of the imaginary-time Schr\"odinger dynamics, when time $t$ is large, Eq.~\eqref{eq:sol1_Gamma} becomes
\begin{equation}
\label{eq:sol1_approx}
\Gamma(t) \sim \frac{M}{\beta (4N)^{1/2N}} (c_1 t+c_2)^{-1/2N}.
\end{equation}
The forms of the two equations, \eqref{polyn_power} and \eqref{eq:sol1_approx}, are quite analogous.
It is difficult to explain intuitively why those two completely different dynamical processes follow almost identical convergence conditions.  This is particularly so, given that the real-time Schr\"odinger dynamics operates at zero temperature whereas the present results, Proposition 1 and Corollary 1, apply at any temperature, high or low, although one usually chooses very low temperatures to study the ground-state behavior of the system by SQA. We cannot exclude the possibility of an accidental coincidence, but also there may lie some deep physics behind this fact. We leave it to future work to try to answer this intriguing question.

\begin{acknowledgments}
This work is based on a project JPNP16007 commissioned by the New Energy and Industrial Technology Development Organization (NEDO). 
\end{acknowledgments}

\appendix
\label{appendix}
\section{Matrix element and operator norm}
We prove the inequality \eqref{eq:norm_inequality}.
Given an operator $\mathcal{A}$, the definition of the operator norm is 
\begin{equation}
\lVert \mathcal{A}\rVert:={\rm sup}_{\lVert v\rVert=1} \lVert\mathcal{A}v\rVert,
\end{equation}
where one should note that $\lVert \cdots \rVert$ on the left-hand side is the operator norm whereas the one on the right hand side is the norm of a vector.
Therefore, for any vectors $v$ and $w$ of norm 1, $\lVert v\rVert=\lVert w\rVert=1$, 
\begin{equation}
\lvert\bra{w}\mathcal{A}\ket{v}\rvert\le \lVert w\rVert\cdot\lVert\mathcal{A}v\rVert=\lVert\mathcal{A}v\rVert\le\lVert\mathcal{A}\rVert.
\end{equation}

\input{output.bbl}

\end{document}

%% file: output.bbl
%